\documentclass[12pt]{article}

\usepackage[utf8]{inputenc}
\usepackage{geometry}
\usepackage{indentfirst}
\usepackage{setspace}
\usepackage{titlesec}
\titlespacing*{\section} {0pt}{1.5ex plus .1ex}{1.5ex plus .1ex}
\titlespacing*{\subsection} {0pt}{1ex plus .1ex}{1ex plus .1ex}
\titlespacing*{\subsubsection}{0pt}{0.75ex plus .1ex}{0.75ex plus .1ex}

\usepackage{float}
\restylefloat{table}
\usepackage{booktabs}
\usepackage{multirow}
\usepackage{array} 
\usepackage{paralist} 
\usepackage{verbatim} 
\usepackage{subfig} 
\usepackage[bottom]{footmisc}

\usepackage[toc,page]{appendix}
\usepackage{graphicx} 
\usepackage{microtype}
\providecommand{\keywords}[1]{\textbf{\textit{Keywords:}} #1}
\providecommand{\JELs}[1]{\textbf{\textit{JEL Classification:}} #1}

\usepackage{amsmath,amsthm,amssymb}
\usepackage{mathtools}
\usepackage{array}
\usepackage{bbm}
\usepackage{dsfont}

\newtheoremstyle{exampstyle}
{\topsep}
{\topsep} 
{} 
{} 
{\bfseries}
{.}
{.5em} 
{} 
\theoremstyle{exampstyle} 
\theoremstyle{exampstyle} 
\theoremstyle{exampstyle} 
\theoremstyle{exampstyle} \newtheorem{prop}{Proposition}
\theoremstyle{exampstyle} 
\theoremstyle{exampstyle} 
\theoremstyle{exampstyle} 

\begin{document}
	
	\title{Blockchain innovation in promoting employment\thanks{We thank Fan Liu, Weibiao Xu, and Ziliang Yu for their helpful comments. All errors are our own.}}
	\date{\vspace{3mm} \today}
	\author{David Lee Kuo Chuen \medskip\\{\normalsize Singapore University of Social Sciences}\medskip\\\textit{\normalsize davidleekc@suss.edu.sg} \and Yang Li\footnote{Corresponding author. 463 Clementi Road, Singapore, 599494.}\medskip\\{\normalsize Singapore University of Social Sciences}\medskip\\\textit{\normalsize liyang@suss.edu.sg}}
	\maketitle
	
	\onehalfspacing
	
	\thispagestyle{empty}
	
	\begin{abstract}
		Blockchain technology, though conceptualized in the early 1990s, only gained practical relevance with Bitcoin's launch in 2009. Recent advancements have demonstrated its transformative potential, particularly in the digital art and global payment sectors. Non-fungible tokens (NFTs) have redefined digital ownership, while financial institutions use blockchain to enhance cross-border transactions, reducing costs and settlement times. Using the Diamond-Mortensen-Pissarides (DMP) model, this paper examines blockchain’s impact on labor markets by improving job-matching efficiency, thereby reducing unemployment. However, high research costs and competition with incumbent technologies hinder early-stage blockchain adoption. We extend the DMP model to analyze the role of government intervention through tax and wage policies in mitigating these barriers. Our findings suggest that lowering firm tax rates can accelerate blockchain innovation, enhance labor market efficiency, and promote employment growth, highlighting the critical balance between technological progress and economic policy in fostering blockchain-driven economic transformation.
	\end{abstract}

	\noindent \keywords{Blockchain, NFTs, Labor Market Efficiency, R\&D Costs,Government Policy}
	\\
    
	\noindent \JELs{E24, O33, G28}
    
	\strut

	\pagebreak
	\setcounter{page}{1}
	
\section{Introduction}
	The basic concept of blockchain was first introduced in early 1990s (Haber and Stornetta, 1991) but there is no real product emerged until Nakamoto wrote his paper (Nakamoto, 2008) and developed Bitcoin in 2009. Although blockchain may not have any real potential, others may do. Recently, blockchain technological advances have shown potential to significantly alter the artwork marketplace and the global payment system.
	
	The NFTs artwork market is growing and give new meaning to digital art. Non-fungible tokens (NFTs) are the newest way to trade collectibles, and encompass art forms ranging from digital art to GIFs to music and beyond. Essentially, they are like any other physical collector's item, but instead of receiving an oil painting on canvas to hang on your wall, for example, you get a JPG file. Because they hold value, they can be bought and sold just like other types of art – and, like with physical art, the value is largely set by the market and by demand. Crucially, NFTs are part of the Ethereum blockchain (Ethereum is a type of cryptocurrency), and the NFT data is stored on the blockchain, each token cannot be destroyed, lost or duplicated. Ownership of these tokens is also immutable, which means collectors actually possess their NFTs, not the companies that create them. Another benefit of storing historical ownership data on the blockchain is that items such as digital artwork can be traced back to the original creator, which allows pieces to be authenticated with no chance of buying a counterfeit art piece. Therefore, if successful, the NFTs would reshape the artwork market place.
	
	In addition, there have been an increasing number of financial institutions that use blockchain technology to disrupt the global payment system, such as Ripple, American Express, UBS, and Standard Chartered. The primary goal of these fintech companies and financial institutions is to make cross-border and cross currency transaction in a cheaper and faster way. As we all known, the speed of sending remittance cross border is very slow and may often take two days while the cost is very high. With the help of blockchain, Ripple uses cryptography and information technology to record data in open and distributed ledgers. This protocol is able to significantly increase the security and efficiency of information transfer and transactions. According to Ripple’s report (2016), it can cut bank’s global settlement costs up to 60\%. This example illustrates that the promising future of utilizing blockchain technology to reshape the global payment system.
	
	To gain some idea about how the success of the above two examples creates new job opportunities effectively, we first follow the Diamond-Mortensen-Pissarides (DMP) model with Cobb-Douglas form matching function and then we shed light on what effects of a technological progress in increasing matching efficiency on unemployment rate in stationary equilibrium. For example, when a NFTs marketplace sells digital artwork or a financial institution initiates a cross border (cross currency) payment using more advanced technology, it is more likely that the matching technology would be improved. This fact may make creating a new vacancy (job) more attractive to firms because the vacancy will be filled more quickly and thereby it may eventually lower the unemployment. On the other hand, the increase in efficiency also implies that vacancies would be filled faster, which tends to reduce the number of vacancies in the economy. In other words, the blockchain innovation is able to promote matches between firms and workers in deceasing both vacancies and unemployment rate.
	
	It is worthy to note that basic concepts have existed for a while, however, technological progress of blockchain-based real world product was very slow. To gain some idea about the small market size of blockchain based product, consider the following. Technological innovations, particularly disruptive ones, often need tremendous research costs and face severe competition with incumbent traditional matured ones. In the early stage of technological progress, if the initial blockchain innovation is not advanced enough, then new discoveries are not able to compensate the research costs. In this case, the incumbent traditional technology cannot be displaced and it still retains the market even if its matching process is relatively less efficient. Although the ultimate level of blockchain innovation would be sufficiently large so as to replace the old products as technological progress continues, it will not be developed once the initial innovations was gave up. 
	
	Besides the effect of improved match efficiency, we next extend our model to investigate the negative effect of research cost on reducing unemployment rate. We find that there is no chance for a blockchain innovation to be developed into real world product since it would not enable to yield enough return to cover the research cost and hire workers in the early stage. This result is in fact consistent with the recent development of blockchain based product. We then modify the wage determination through Nash Bargaining in the DMP model and assume the government mandates that the wage is equal to a fraction of the firm's pre-tax output. The government is also supposed to be able to impose a tax on both firms and workers. In this modified model, we can investigate whether and how the tax rate and the wage determination mitigate the negative effects of research costs on innovation. Intuitively, a decrease in firm-tax rate tends to help technological progress and improve social welfare in the following two aspects. It brings the frontier of knowledge to the required critical level to yield higher returns, but also increases the employment rate by stimulating inventors engaged in R\&D.
	
	\textbf{Related literature.} Blockchain technology has attracted growing academic interest due to its potential to disrupt traditional markets, particularly in financial transactions and digital asset ownership. This section reviews key contributions related to blockchain’s economic impact, its role in labor markets, and the challenges of early-stage technological adoption.

    \textbf{Blockchain Innovation and Market Disruption.} Nakamoto (2008) developed Bitcoin, demonstrating blockchain’s real-world applicability. Since then, research has explored its broader economic implications. Catalini and Gans (2016) discuss how blockchain reduces verification and networking costs, enabling decentralized trust mechanisms. Böhme et al. (2015) provide an overview of Bitcoin’s evolution and its role in financial transactions, highlighting its disruptive potential in payments. Further, Cong and He (2019) analyze how blockchain’s transparency and immutability improve contract efficiency and reduce agency problems in financial markets.

    \textbf{NFTs and the Digital Art Market.} The emergence of non-fungible tokens (NFTs) has sparked discussions on digital ownership and value creation. Kher et al. (2022) examine the economic mechanisms behind NFTs, emphasizing their role in establishing scarcity and ownership in digital markets. Dowling (2022) highlights the speculative nature of NFT pricing, noting that market dynamics resemble traditional art markets. Another key contribution by Nadini et al. (2021) provides empirical evidence on the structure of NFT transactions, revealing patterns in pricing, resales, and collector behavior. These studies collectively suggest that NFTs have reshaped digital art markets by creating new revenue streams for artists and enhancing transparency in provenance tracking.

    \textbf{Blockchain in Financial Services.} Blockchain’s impact on the global payment system has also been extensively studied. Ripple, a prominent blockchain-based payment network, has demonstrated significant efficiency gains in cross-border transactions (Ripple, 2016). Lee and Deng (2018) examine how blockchain improves remittance processes, reducing both transaction time and cost. Fernández-Villaverde and Sanches (2019) extend this discussion by modeling the monetary implications of cryptocurrencies and decentralized financial systems. These studies support the argument that blockchain-based financial applications can enhance transaction security and liquidity, disrupting traditional banking models.

    \textbf{Blockchain and resilience.} The growing reliance on blockchain-based financial infrastructure underscores the need to enhance its resilience and reliability, particularly in asset management and financial stability. Cheng et al. (2021) offer a resilience assessment framework applicable to blockchain networks in mitigating systemic risks, while their work on multi-hazard disruptions informs blockchain defense against cyberattacks and transactional failures. Reliability modeling (Cheng and Elsayed, 2017 and 2018), can be extended to smart contracts and decentralized applications to minimize failure probabilities. Hybrid blockchain models integrating permissioned and permissionless mechanisms may benefit from insights on mixed-unit reliability (Cheng and Elsayed, 2016). Resilience-based restoration strategies (Cheng et al., 2024) could shape blockchain policies amid market stress, while maintenance frameworks for self-service systems (Wei et al., 2025) can support DeFi liquidity stability. The shift from reliability to resilience (Cheng et al., 2024) underscores the need for proactive measures beyond fault tolerance. Additionally, network analysis techniques used to assess financial shock transmission (Li et al., 2024) can be applied to blockchain consensus mechanisms to evaluate systemic risks. Optimal sequential testing plans (Cheng and Elsayed, 2020) and fault diagnosis methods under varying conditions (Gao et al., 2021) may enhance blockchain security and efficiency. Furthermore, studies on critical infrastructure resilience (Cheng et al., 2021) and service reliability in self-service systems (Wei et al., 2024) provide valuable insights into maintaining blockchain network robustness. Integrating these resilience strategies with financial stability models can strengthen blockchain-driven asset management against liquidity crises and panic-driven redemptions, complementing broader financial stability research (Allen and Gale, 1998; Diamond and Dybvig, 1983; Acharya and Yorulmazer, 2008).

    \textbf{Technological Innovation, Labor Markets, and Policy Implications.} The broader economic implications of blockchain adoption on labor markets can be analyzed through the Diamond-Mortensen-Pissarides (DMP) framework (Diamond, 1982; Mortensen and Pissarides, 1994). This model provides insights into how technological progress influences job creation and unemployment through improved matching efficiency. Acemoglu and Restrepo (2018) explore the impact of automation on labor markets, arguing that technology-induced efficiency gains can reduce unemployment in some sectors while displacing jobs in others. Similarly, Aghion et al. (2019) investigate the role of innovation incentives in shaping labor demand and firm productivity. However, blockchain adoption faces high research and development (R\&D) costs, limiting its initial growth. Bresnahan and Trajtenberg (1995) discuss how general-purpose technologies require significant investment before reaching widespread adoption. Jovanovic and Rousseau (2005) emphasize that early-stage technological advancements often struggle to compete with established incumbents, delaying diffusion. To address these barriers, Bloom et al. (2013) examine the role of government intervention in fostering innovation through tax incentives and R\&D subsidies. Their findings suggest that reducing firm tax rates can accelerate technological adoption and labor market expansion.

    While previous studies provide valuable insights into blockchain’s disruptive potential, the interaction between blockchain innovation, labor market dynamics, and policy interventions remains underexplored. This paper builds on the existing literature by incorporating a modified DMP framework to assess how blockchain’s impact on job-matching efficiency can mitigate unemployment, despite high research costs. Additionally, we extend previous findings by analyzing how government tax and wage policies influence blockchain adoption and employment growth.
	
	In Section 2, we describe our basic model in detail and the results of unemployment rate in stationary equilibrium; and the analysis of government tax rate policy is given in Section 3. Section 4 concludes our insightful results.
	
	\section{Basic model}
	We will begin by describing the economic agents in the Diamond-Mortensen-Pissarides (DMP) model and the environment in which they operate and then incorporate the matching efficiency characterized by technology innovation associated with research cost.

	\subsection{The Environment}
	The economy is consist of a continuum of workers, \(i, \in [0,1]\) in infinite periods with \(t=0, 1, 2, ...\). In each period, a worker is either employed and thus earns wage \(w_t\) or unemployed and thus receives benefit \(b\). Note that \(w_t\) is a market wage in this model and that all employed workers will be earning the same wage. In other words, there is not going to be any process of searching for a better job offer in our model. Workers are risk neutral and the worker's discount factor is labeled as \(\beta\) that is equal to \(\beta \equiv 1/(1+r)\) where \(r\) is the net interest rate. In addition, we assume that an employed worker loses job with probability \(\sigma > 0\) in each period. 

	In the other side of the labor marker, a firm's job (one firm provides only one job opportunity), in each period, is either filled and thus produces output \(y\) (i.e., the firm earns profit \(y-w_t\)) or vacant and thus produces nothing (i.e., the firm pays vacancy cost \(c>0\)). We assume the product marker is competitive and the assumption of free entry and exit implies that (1) new firms can create vacant jobs and (2) existing firms can exit to avoid vacancy cost \(c\). All firms are risk neutral and the discount factor for them is also equal to \(\beta\) as workers. In other words, the firms in this model maximize the present value of their expected profits. Note that we assume the filled job produces value such that the productivity yields higher return than the unemployed benefit (i.e., \(y>b\)).
	
	We use the matching function approach pioneered by Diamond-Mortensen-Pissarides to model the process in which it is difficult for unemployed workers to find firms with vacancies. Specifically, this aggregate matching function takes the Cobb-Douglas form
	\[
		M(u_t,v_t) = A u_t^\alpha v_t^{1-\alpha},
	\]
where \(M\) is number of unemployed workers who find a job, \(u_t\) is the number of workers looking for a job, and \(v_t\) is the number of firms with vacant jobs. This matching technology (like a production function) summarizes the outcome of a complex process of how many successful matches of an unemployed worker to a vacant job will occur, given the number of workers looking for a job and given the number of available jobs out there. It is straightforward to show that the first partial derivatives with respect to \(u_t\) and \(v_t\) are positive respectively, that is, the more workers are looking jobs and the more jobs will be filled, and the more job are available and the more workers will find one. This matching form also shows that the marginal products are diminishing as we expected. Finally, the function tells us it is constant returns to scale, for example, if the labor market double in size (twice as many workers looking for jobs and twice as many firms trying to hire), we will get twice as many successful matches.

	As in DMP model, it is useful to define some notation based on the matching function. The labor market tightness, defined as \(\theta_t \equiv v_t/u_t\), is a word traditionally used to describe labor markets; the labor market is said to be tight if, from a firm's point of view, it is difficult to find workers. In this sense, we will say that the market is ``tight'' of there are many vacancies relative to the number of unemployed workers, while the labor market is ``slack'' if there are few vacancies and many unemployed workers. The job finding rate, defined as \(p(\theta_t) \equiv M(u_t,v_t)/u_t = M(1,\theta_t) = A \theta_t^{1-\alpha}\), is the probability that an employed worker will find a job. Note that this job find rate \(p(\theta_t)\) is an increasing and concave function. When the labor market is tighter, it is easier for a work to find a job, but there are diminishing returns to this process. The vacancy filling rate, defined as \(q(\theta_t) \equiv M(u_t,v_t)/v_t = M(1/\theta_t,1) = A \theta_t^{-\alpha} \), is the probability that a given vacancy will be filled. Note that this rate is a decreasing and concave function. When the labor market is tight, a firm with a vacancy is less likely to find a worker, but there are diminishing returns to this process again.
	
	\subsection{Innovation and matching efficiency}
	Now we consider the firm is able to improve the efficiency of the matching process measured by an increase in \(A\) in the Cobb-Douglas matching function through blockchain technology innovation. It is worth to noting that this improvement will cost the firm \(\eta\) units. In addition, we assume that the efficiency measure \(A(\eta)\) is function of the research cost \(\eta\) and it is an increasing and concave function. In this subsection, we will next solve the model focusing on the stationary equilibrium and investigate whether achieving an advanced blockchain technology is able to compensate the research cost and thus reduce vacancies and unemployment rate compared to the status with no technology innovation.
	
	\subsection{Steady-state equilibrium}
	\textbf{Beveridge curve.} In the first step in solving the model, we look at the Beveridge curve. We can use the ``lake'' diagram to capture how workers move back and forth between employment and unemployment. The number of workers moving from unemployment to employment depends on labor market conditions, which reflect how many firms (jobs) are hiring. For any given value of tightness \(\theta_t\), the dynamics of the unemployment rate is determined by
	\[
		u_{t+1} = u_t + \delta (1-u_t) - p(\theta_t) u_t.
	\]
In general, note that we can think of the value of a match between a firm and a worker as being stochastic over time. For example, the demand for the good (especially the technology product) produced by the firm may change. If the value falls too much, the firm may prefer to close and lay off the worker. The specification here is a simple representation of that process. We can think that with probability \((1-\delta)\), the demand for the firm's product stays the same, but with probability \(\delta\) the demand completely disappears (perhaps another technology product becomes more popular among consumers). 

	Our focus in the model is going to be on stationary equilibria, where \(u_t\) is constant over time. It is easy to show that a steady-state of the lake model is characterized by:
	\[
		u = u + \delta (1-u) - p(\theta) u \text{ or } u = \frac{\delta}{\delta + p(\theta)}.
	\]
Substituting the job finding rate into the above equation, we obtain one of our key equations, we label it as ``Beveridge Curve'' (BC):

	\begin{equation}
		\label{BC}
		u = \frac{\delta}{\delta + A(\eta) \theta^{1-\alpha}} \tag{BC}
	\end{equation}	

	This equation tells us the steady-state relationship between vacancies \(v\) and unemployment rate \(u\). To understand the level of unemployment, however, we need to look at firms' incentives to create new vacancies, that is, the process of job creation.

	\textbf{Job creation.} We again focus on steady-state allocations and in each period, a firm is in one of two states: job filled (F) and job vacant (V). Given constant \(u\) and \(v\), the Bellman equations that characterize the expected present value of profits for a firm that is in each of these two states are given by:
	\begin{align*}
		J_F &= y - w - \eta + \beta [\delta J_V + (1-\delta) J_F] \\
		\text{and } J_V &= -c + \beta [q(\theta) J_F + (1-q(\theta)) J_V],
	\end{align*}
where \(J_F\) is the value of having a filled job and \(J_V\) is the value of having a vacant job. When the job is filled, the profit for a firm is equal to the productivity \(y\) minus the wage \(w\) and the research cost \(\eta\) (note here that if the firm refuses to update its technology level and thus the research cost \(\eta=0\) and the measure of matching efficiency stays at \(A\)); the profit for a firm tomorrow discounted by \(\beta\) is equal to the expected profits of losing worker with probability \(\delta\) and keeping worker with probability \(1-\delta\). When the job is vacant, the firm suffers vacancy cost \(c\) and the profit tomorrow discount by \(\beta\) is equal to the expected profits of finding a worker with probability \(q(\theta) = A(\eta) \theta^{-\alpha}\) and being vacant with probability \(1-q(\theta)\).

	Notice that here we are using the fact that a firm that loses its worker will choose to create a new vacancy.The endogenous variables in the model \((w, \theta)\) must adjust so that this behavior reflects an optimal choice. We are also using the fact that we are looking at a stationary equilibrium in writing \(J_F\) and \(J_V\) as being independent of time. Notice also that the firm takes the wage that it must pay to the worker as given. In particular, the firm is not searching to find a worker who will accept a lower wage since our primary objective in this paper is to figure out the effects of technology innovation on unemployment rate. In other words, we simplify the on-the-job search process and every worker will need to be paid the ``market'' wage, and any worker the firm is matched with is just as good as another.

	Before solving the above two Bellman equations, we first clarify that free entry condition yields \(J_V=0\) always holds in a steady-state equilibrium, which helps us solve these equations. Suppose that \(J_V>0\), that is, if there were positive gains to entry, all potential firms would enter (i.e., \(v=\inf\)). But then the labor market would be so tight (i.e., \(\theta = \inf\)) that firms would find it impossible to hire a worker (i.e., \(q(\theta)=0\)). If it is impossible to hire a worker, then the value entry cannot be positive (i.e., \(J_V = -c/(1-\beta)\)). Thus, it is a contradiction. It is actually possible to have \(J_V<0\) if the vacancy cost \(c\) is very large. But in this case, no jobs are ever created (i.e., \(v=0\)) and thus there is no match (i.e., \(M(u,v)=0\)) nor impossible to find a job (i.e., \(p(\theta)=0\)). This fact gives us the allocation has all workers being unemployed (i.e., \(u=1\)), which is not so interesting. Therefore, when we say that we focus on steady-state allocations with \(0<u<1\), we are implicitly assuming that \(c\) is not too large. We will refer to \(J_V=0\) as the ``free entry condition'' for this model.
	
	Plugging \(J_V=0\) into the second Bellman equation, the free-entry condition pins down the value of being a firm with a worker in equilibrium is determined by:
	\[
		\beta J_F = c \frac{1}{q(\theta)} = c \frac{1}{A(\eta) \theta^{-\alpha}}.
	\]
This equation is intuitive. The expected profit from a filled job (starting tomorrow) must exactly offset the expected cost of finding a worker to hire where \(c\) is the finding cost per period and \(1/q(\theta)\) is the expected duration of vacancy.

	Combining the above equation, \(J_V=0\) and the first Bellman equation, we have
	\[
		\frac{c}{\beta q(\theta)} = y - w - \eta + \beta (1-\delta) \frac{c}{\beta q(\theta)}.
	\]
Using \(\beta = 1/(1+r)\), we have
	\begin{equation}
	\label{JC}
	w = y - \eta - \frac{r+\delta}{A(\eta) \theta^{-\alpha}} c. \tag{JC}
	\end{equation}
This is going to be our another key equation, labeled as ``Job Creation'' condition (JC). This equation, for a given level of labor market tightness \(\theta\), tells us what value of the wage will make a firm exactly indifferent between creating a new vacancy and not. In addition, this relationship is like a demand curve for labor. When the wage is high, few firms will create jobs and, as a result, there will be a lot of slack in the labor market. In this sense, the demand for labor will be low. When the wage is low, in contrast, many firms will create new jobs and the demand for labor coming from these firms with new openings will be high.

	\textbf{Wage determination.} Unlike the DMP model based on bargaining theory, our approach assumes that the government passes a law stating that every worker's wage must be equal to a fraction \(\gamma\) of the firm's output \(y\) after paying research cost \(\eta\). In other words, in our modified model there is no Nash bargaining. Instead, each employed worker's wage is mandated by law to be \(w = \gamma (y - \eta)\), where \(\gamma\) is a number set by the government. Notice that this modification in our model plays important roles in the following two aspects. Firstly, it can help us how this parameter \(\gamma\) affects equilibrium outcomes and thus understand whether and how this wage intervention improves allocations and reduces unemployment rate. Secondly, this wage rule avoids determining the wage through utilizing the measure of bargaining power that is difficult to calibrate in reality.
	
	If the government sets mandated wage as \(w = \gamma (y - \eta)\), this condition together with the job creation condition (JC) determines the tightness \(\theta^*\) in equilibrium:
	\[
		\gamma (y - \eta) = y - \eta  - \frac{r+\delta}{A(\eta) (\theta^*)^{-\alpha}}
	\]
or,
	\[
		\theta^* = \left[ \frac{(1-\gamma) A(\eta) (y-\eta)}{(r+\delta) c} \right]^{\frac 1\alpha}.
	\]
Combining with the Beveridge curve (BC) and the tightness \(\theta^*\), we have the unemployment rate in equilibrium \(u^*\) characterized as
	\[
		u^* = \frac{\delta}{\delta + p(\theta^*)}
	\]
or,
	\begin{equation}
		\label{u_base}
		u^* = \frac{\delta}{ \delta + \left[ \frac{1-\gamma}{(r+\delta)c} \right]^{\frac{1-\alpha}{\alpha}} \left[ A(\eta) (y-\eta)^{1-\alpha} \right]^{\frac 1\alpha} }.
	\end{equation}
	
	\subsection{Effects of innovation and wage rule on unemployment}
	Now that we have already derived the steady-state unemployment rate of the model, we can investigate its comparative statics. How will an innovation associated with an increase in research cost \(\eta\) affect the unemployment? Does blockchain innovation promote employment compared to the equilibrium outcome with a less developed technology? We will start with the effect of innovation on unemployment.
	
	Too see how \(\eta\) changes the steady-state equilibrium value \(u^*\), let us differentiate equation \eqref{u_base} with respect to \(\eta\). Then, we have
	\[
		\frac{d u^*}{d \eta} = - \frac{ \delta \left[ \frac{1-\gamma}{(r+\delta) c} \right]^{\frac{1-\alpha}{\alpha}} \frac 1\alpha \left[ A(\eta)(y-\eta)^{1-\alpha} \right]^{\frac 1\alpha -1} }{ \left\{ \delta + 	\left[ \frac{1-\gamma}{(r+\delta) c} \right]^{\frac{1-\alpha}{\alpha}} \left[ A(\eta)(y-\eta)^{1-\alpha} \right]^{\frac 1\alpha} \right\}^2 } \cdot (y-\eta)^{-\alpha} \left[ A'(\eta) (y-\eta) - (1-\alpha) A(\eta) \right].
	\]
We can actually see this sign of \(d u^* / d \eta\) by looking at the last term of the above equation \( A'(\eta) (y-\eta) - (1-\alpha) A(\eta) \). After differentiating this term with respect to \(\eta\), we are able to show that it is decreasing in \(\eta\). In other words, the sign of \(d u^* / d \eta\) is negative when the research cost \(\eta\) is relatively low while it is positive when \(\eta\) is very high. Therefore, the effect of innovation on unemployment rate is ambiguous, which gives us the next result.

	\begin{prop}
		There exist a threshold value of \(\hat{\eta}\) such that when achieving an innovation at research cost below \(\hat{\eta}\) then it will reduce unemployment in equilibrium; otherwise an innovation increases unemployment in steady-state.
	\end{prop}
	
	\begin{proof}
		First, we take derivative of \( A'(\eta) (y-\eta) - (1-\alpha) A(\eta) \) with respect to \(\eta\). This yields \(A''(\eta) (y-\eta) - (2-\alpha) A'(\eta)\) that is negative since we assume that \(A(\eta)\) is an increasing and concave function and that \(\eta < y\) and \( 0<\alpha<1\). Secondly, when \(\eta\) is equal to 0, we have \( A'(0) (y-0) - (1-\alpha) A(0) >0\); while when \(\eta\) is equal to \(y\), we have \( A'(y) (y-y) - (1-\alpha) A(y) <0\). Thus, the term \( A'(\eta) (y-\eta) - (1-\alpha) A(\eta) \) uniquely pins down to zero at \(\eta = \hat{\eta}\). Finally, it is easy to show that the term \( A'(\eta) (y-\eta) - (1-\alpha) A(\eta) > 0 \) as \(\eta < \hat{\eta}\), which implies that \(d u^* / d \eta <0\) (i.e., an innovation associated a lower research cost promotes unemployment). A larger increase in \(\eta\) such that \(\eta > \hat{\eta}\), however, the innovation causes more unemployment undesirably.
	\end{proof}

The intuition behind this result is relatively straightforward because there are two competing effects. The development in technology associated with larger research cost in \(\eta\) tends to promote the match efficiency \(A'(\eta)>0\), which encourages firms to create more vacancies and makes workers easier to find a job. However, this increase in \(\eta\) tends to reduce firm's profits \(y-\eta\), which causes firms to create fewer vacancies. We can see that the positive effect dominates if the cost is small, however, the negative effects turns out to be dominant as the research cost is very large. Figure 1 illustrate the above result with \(A(\eta) = A \cdot (\eta-\frac 12 \eta^2)\) and the parameter set (\(A=1, y=1, b=0, c=1, r=0.5, \delta=0.5, \alpha=0.5, \gamma=0.5\)) as a numerical example.

	\begin{figure}[H]
		\centering
		\includegraphics[width=\textwidth]{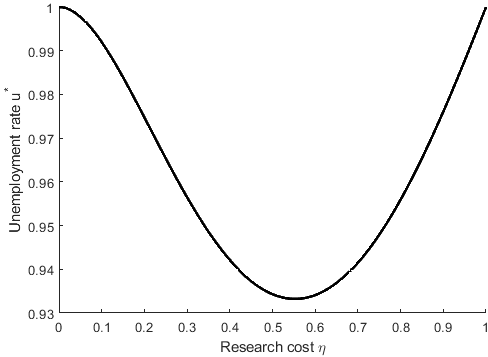}
		\caption{The effect of innovation on unemployment rate in equilibrium}
	\end{figure}
	
	By focusing on right region where the research cost is very high, we can see that an innovation in blockchain technology will cause an increase in unemployment rate. This fact gives a concern that when the innovation needs large research cost the development of technology in fact causes undesirable result of severed unemployment. Therefore, this result leaves an open question whether the government is able to promote employment even if the research cost is very high by adjusting the wage rule \(w = \gamma (y-\eta)\) effectively.
	
	We again start this analysis by differentiating the unemployment rate \(u^*\) given in equation \eqref{u_base} with respect to \(\gamma\), which yields
	\[
		\frac{d u^*}{d \gamma} = \frac{ \delta \frac{1-\alpha}{\alpha} \left[ \frac{1-\gamma}{(r+\delta) c} \right]^{\frac{1-\alpha}{\alpha}} \left[ A(\eta)(y-\eta)^{1-\alpha} \right]^{\frac 1\alpha} }{ \left\{ \delta + 	\left[ \frac{1-\gamma}{(r+\delta) c} \right]^{\frac{1-\alpha}{\alpha}} \left[ A(\eta)(y-\eta)^{1-\alpha} \right]^{\frac 1\alpha} \right\}^2 } \cdot \frac{1}{1-\gamma} > 0.
	\]	
Thus, this result captures the idea of government support on innovation through a decrease in wage, which can be interpreted as an increase in firm's profits. The next proposition states this result and it is depicted in Figure 2 by varying \(\gamma = 0.1, 0.5, 0.9\).

	\begin{prop}
		When the government mandates a lower wage, this decrease in \(\gamma\) is able to promote employment rate in equilibrium even if the research cost is very high.
	\end{prop}

	\begin{figure}[H]
		\centering
		\includegraphics[width=\textwidth]{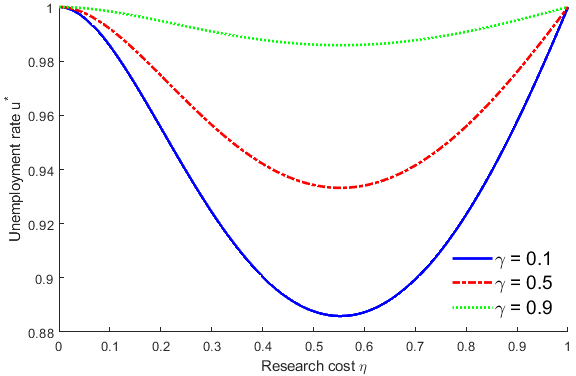}
		\caption{The effect of wage rule on unemployment rate in equilibrium}
	\end{figure}
	
	\section{Tax policy}
	In the previous section, we show that a reduce in wage rule \(\gamma\) is able to improve employment even if the corresponding research cost is very high. However, it is difficult to adjust the wage level in practice. Given the high research cost, is it still possible to reduce unemployment rate in equilibrium? In this section, we further incorporate a tax policy into the basic model and investigate the impact of tax rate on unemployment rate.
	
	\subsection{Equilibrium outcomes under the tax regime}
	Now consider the government imposes a tax on both firms and workers. In particular, each firm must pay a fraction \(\tau_f\) of its profit to the government and each worker must pay a fraction \(\tau_w\) of her wage to the government. After taxes, the firm keeps \((1-\tau_f)(y-\eta-w)\) and the worker keeps \((1-\tau_w)w\). The wage is still required to satisfy \(w=\gamma(y-\eta)\). 
	
	Note that after-tax profits in this model are \((1-\tau_f)(y-\eta-w) = (1-\tau_f)(1-\gamma)(y-\eta)\). Replacing \((1-\gamma)(y-\eta)\) from the job creation condition \eqref{JC} into this expression, we have new tightness \(\theta_\tau^*\) in equilibrium, which is determined by
	\[
		\theta_\tau^* = \left[ \frac{(1-\tau) (1-\gamma) A(\eta) (y-\eta)}{(r+\delta) c} \right]^{\frac 1\alpha}.
	\]
Plugging this new measure into the Beveridge curve (BC), we have the revised unemployment rate as follows.
	\begin{equation}
		\label{u_tax}
		u_\tau^* = \frac{\delta}{ \delta + \left[ \frac{(1-\tau_f) (1-\gamma)}{(r+\delta)c} \right]^{\frac{1-\alpha}{\alpha}} \left[ A(\eta) (y-\eta)^{1-\alpha} \right]^{\frac 1\alpha} }.
	\end{equation}
	
	\subsection{Effect of tax rate on unemployment}
	To see how \(\tau_f\) changes the steady-state unemployment rate, we differentiate equation \eqref{u_tax} with respect to \(\tau_f\). Then we have
	\[
		\frac{d u_\tau^*}{d \tau_f} = \frac{ \delta \frac{1-\alpha}{\alpha} \left[ \frac{(1-\tau_f) (1-\gamma)}{(r+\delta) c} \right]^{\frac{1-\alpha}{\alpha}} \left[ A(\eta)(y-\eta)^{1-\alpha} \right]^{\frac 1\alpha} }{ \left\{ \delta + 	\left[ \frac{(1-\tau_f) (1-\gamma)}{(r+\delta) c} \right]^{\frac{1-\alpha}{\alpha}} \left[ A(\eta)(y-\eta)^{1-\alpha} \right]^{\frac 1\alpha} \right\}^2 } \cdot \frac{1}{1-\tau_f} > 0.
	\]
Therefore, decreasing the tax lowers the unemployment rate in equilibrium. Intuitively, a decrease in tax yields higher firm's profits, which increase the value of a filled position. As a result, more vacancies will be created and to make the labor market to be tight. This fact make unemployment spells shorter on average, which reduces the steady-state unemployment rate. The next proposition states this analytical result and it is depicted in Figure 3 by varying \(\tau_f = 0, 0.5, 0.9\) given \(A(\eta) = A \cdot (\eta-\frac 12 \eta^2)\) and the parameter set (\(A=1, y=1, b=0, c=1, r=0.5, \delta=0.5, \alpha=0.5, \gamma=0.5\)).

	\begin{figure}[H]
		\centering
		\includegraphics[width=\textwidth]{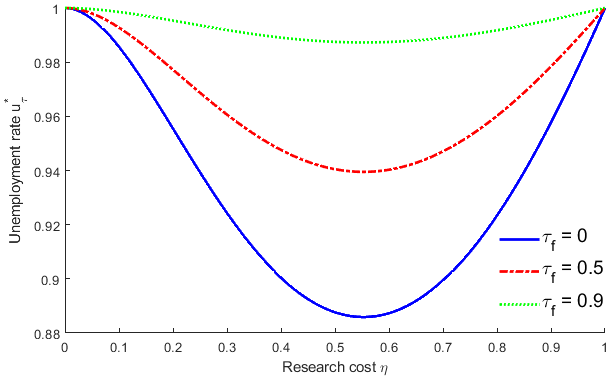}
		\caption{The effect of tax rate \(\tau_f\) on unemployment rate in equilibrium}
	\end{figure}

	\begin{prop}
		When the government reduces the tax rate \(\tau_f\) imposed on firm's profits, this policy is able to promote employment rate in equilibrium even if the research cost is very high.
	\end{prop}

    \section{Conclusion}
    Blockchain resilience is critical for maintaining financial stability, particularly in asset management, where systemic risks can lead to liquidity crises and panic-driven redemptions. By integrating resilience assessment frameworks, network analysis techniques, and financial stability models, blockchain-based systems can better withstand cyber threats, transactional failures, and market shocks. Leveraging insights from reliability engineering and financial crisis research, this study highlights the importance of proactive resilience strategies to ensure the robustness and long-term viability of decentralized financial infrastructure.

    Future research should explore the dynamic interplay between blockchain resilience and financial stability by developing real-time monitoring tools for systemic risk detection. Additionally, incorporating machine learning techniques into blockchain security models could enhance anomaly detection and threat mitigation. Further studies on regulatory frameworks and their impact on blockchain resilience would provide valuable insights for policymakers and financial institutions.

    \nocite{*}
    \bibliographystyle{plain}
    \bibliography{bibref}
    
	\end{document}